\documentclass[letterpaper,twocolumn,10pt, conference]{ieeeconf}

\IEEEoverridecommandlockouts
\newcommand{\journal}{0}

\usepackage{graphicx,setspace}
\usepackage{amsmath,amssymb}
\usepackage{color}
\usepackage[english]{babel}
\usepackage{mathrsfs,dsfont}
\usepackage{array}

\usepackage{mathrsfs}
\usepackage[mathscr]{euscript}

\newtheorem{theorem}{\bf Theorem}
\newtheorem{lemma}[theorem]{\bf Lemma}

\newtheorem{assumption}{Assumption}

\newtheorem{definition}{Definition}

\def\QED{~\rule[-1pt]{5pt}{5pt}\par\medskip}
\renewenvironment{proof}{{\bf Proof: \ }}{ \hfill \QED}

\usepackage{tikz,makecell}

\usepackage{etoolbox}

\newcommand{\mytest}[2]{%
  \ifboolexpr{test {\ifnumcomp{\journal}{=}{1}}}
    {#1}
    {#2}
    }

\let\ALP  \mathcal

\newcommand{\beq}[1]{\begin{align} #1 \end{align}}
\newcommand{\beqq}[1]{\begin{align*} #1 \end{align*}}
\renewcommand{\Re}{\mathbb{R}}

\newcommand{\ex}[1]{\mathds{E}\left[#1\right]}
\newcommand{\pr}[1]{\mathds{P}\left\{#1\right\}}

\allowdisplaybreaks
\usepackage[colorlinks]{hyperref}

\title{Auctioning Electricity under Deep Renewable Integration using a Penalty for Shortfall}
\author{Balsam Dakhil \and Abhishek Gupta \thanks{Balsam Dakhil and Abhishek Gupta are with the Department of Electrical and Computer Engineering at The Ohio State University, Columbus, OH, USA. Email:  {\tt\small dakhil.2@osu.edu, gupta.706@osu.edu}. The authors would like to thank NSF ECCS Grant 1610615 for supporting this research.}}

\begin{document}
\maketitle
\thispagestyle{plain}
\pagestyle{plain}
\begin{abstract} 
We analyze the problem of a renewable generator who wants to sell its random generation in a two-stage market to a number of strategic, but flexible, load serving entities (LSEs). To offer an incentive to participate in the auction, the generator promises to pay a penalty associated with any shortfall in generation due to the uncertain nature of the traded resource. We devise an auction mechanism that efficiently allocates electricity among LSEs while eliciting their true valuations. We show that the consumer surplus of the LSEs and expected profit of the generator are positive, thereby creating a win-win situation for all market participants.
\end{abstract}

\section{Introduction}
Renewable energy presents a clean, economic, and environment friendly alternative to traditional sources of energy. It is also thought of as a sustainable source that can solve, or at least hedge against, any possible energy crisis that might happen due to limited fossil fuel reserves. Yet, as much benefits as renewables bring to the world, integrating renewable energy in the existent power grid introduces new challenges to the engineering and operation of the  power system \cite{bird2013integrating}. One major issue that requires attention is the need to design a trading institute that can accommodate for the variability and uncertainty of supply associated with renewable energy generation. The current electricity market is designed for supplying the electricity demanded using fossil fuel and nuclear energy based generators. Existing markets can easily absorb renewable generation at small quantities. However, when the random renewable sources provide substantial portion of the demand, then the existing market structure needs to be changed to ensure reliability and smooth operations of the power grid.

Accordingly, in this paper, we devise an auctioning format for renewable generation that builds on the current structure of electricity market. We consider a renewable power generator selling renewable electricity to $N$ load serving entities (LSEs), or directly to consumers, in a two-stage electricity market. In the day ahead market, LSEs report their inverse demand function and the cost of buying electricity at the spot price should there be a shortfall in supply. Accordingly, the generator contracts to sell a specific amount of electricity to each LSE. In the real time market, the renewable generation is realized, and some LSEs might observe a shortage in supply. As an incentive to motivate LSEs to buy random energy, the generator offers to compensate LSEs that experience a shortfall in the amount of electricity contracted for in the day ahead market. We denote consumers who accept this deal to be flexible. With this compensation in mind, the generator will incur some random cost and has to decide on how much electricity to assign to each LSE and, also, to specify how much each LSE has to pay. In this decision process, the goal is to maximize the social value generated while accounting for the expected cost of compensation resulting from the possibility of not generating the contracted power due to the natural uncertainty in generation. This need for a reliable decision under the inherent uncertainty calls for the the use of stochastic programming with recourse. The employment of this powerful mathematical tool to arrive at an efficient allocation, which maximizes the welfare of all the parties involved in this transaction is one major contribution of this paper. The other important component of this work is the payment scheme and the associated bound on the generator's profit. Before delving in the details of our work, we review some of the relevant literature.

\subsection{Prior Work}
Current power systems were designed to generate, transmit and distribute a fixed supply of electricity to consumers whose demand exhibits some degree of uncertainty. Integrating renewable energy into this system creates new challenges as it introduces uncertainty to the supply end of the system, a feature the current power systems are not equipped to handle \cite{bird2013integrating}. References \cite{decesaro2009wind} and \cite{wind2010solar}, and the references therein, provide an overview of proposed modifications to the operation and control of power systems that are needed to accommodate renewables and fully harness its benefits.
One important aspect of integrating renewables is how to trade renewable energy in the existing electricity market and how this trade impacts its operation. During the last decade, the problem of designing market mechanisms for renewable energy received a lot of attention resulting in large literature. The three main strands of related work study the important problems of minimizing the imbalance costs, distributing the imbalance costs among market participants, and devising novel market mechanisms to manage the imbalance. 

{\it Managing the imbalance costs:} Uncertainty in generation, associated with renewable energy, adds to the challenge of maintaining the reliability of the power grid. In the presence of renewable generation, it is more likely for generators to deviate from generation schedules that are designed to meet the demand in real time, resulting in imbalance costs that reduce the welfare of players in the system. In this setting, papers \cite{morales2010short,bitar2012,skajaa2015}  and \cite{martin2015} studied techniques that minimizes the cost of imbalance, that the generator faces, by computing optimal renewable energy offering that accounts for uncertainty in generation and market prices or engaging in intra-day markets that allow the generator to modify its day ahead offering as the delivery time approaches and better forecasts become available. Another important technique to manage grid stability is harnessing the flexibility of demand to better maintain the balance of the grid. To this end, \cite{albadi2008} surveys flexible resources in electricity market. It did not take long to recognize the importance of demand response in the process of renewable energy integration \cite{DoE}. Many pricing schemes and incentive mechanisms were developed to exploit demand side flexibility, \cite{bitar2013}, \cite{liu2014}, \cite{jia2016} and \cite{gupta2015} sample this important line of research.

{\it Distributing the imbalance costs:}  Meyn et al. \cite{meyn2010} studied the competitive equilibrium of the stochastic two-stage electricity market in the presence of wind  generators and computed the surplus of both suppliers and consumers. They showed that, with deeper levels of penetration, variability in generation, which requires more reserves to meet reliability constraints, only reduces the surplus of consumers but not renewable generators. Based on this observation, they proposed that generators carry the burden of obtaining reserves in the day ahead market by making arrangements with conventional generators or investing in storage. In the mean time, whenever imbalance costs can not be completely avoided and to reduce the financial burden currently placed on consumers, \cite{lin2014} proposes a cost-sharing mechanism which fairly divides the cost of imbalance among the renewable generators that fall short of their contracted amounts.
  
{\it Energy offering strategies:} Trading renewable generation requires new market mechanisms that cater for the stochastic nature of the traded resource. Bitar et al. \cite{bitar2012} explicitly find a formula for the optimal offering and quantitatively relates profitability of one renewable generator to the uncertainty in generation. In \cite{tang2011}, Tang and Jain take one step farther by considering an auction where multiple renewable generators can sell their production to an aggregator, who acts as an agent between the generators and electricity market. 

So far, the market is assumed to be competitive  so that generators are price takers. Recently, the research community became interested in understanding opportunities to exercise market power by paying attention to the strategic behavior of the parties involved in the market and device methods to minimize or even eliminate such instances. Tang and Jain show in \cite{tang2015market} that the Locational Marginal Pricing mechanism opens room for manipulation and propose a second-price-auction-like mechanism for renewable generators to bid their probability distributions in the electricity market. \cite{xu2017efficient} addresses this concern in a setting where conventional generators are anticipated to game the system by modifying their bids to increase their surplus, a situation that reduces the social efficiency of the market. 

In contrast to the existing body of research, this paper promotes load flexibility to consumers who are open to the possibility of not consuming electricity at all if they are compensated accordingly. Consider, for example, an electric vehicle (EV) that is 75\% charged and is plugged in for 2 hours at a commercial charging station (say at a mall or workplace garage). If the EV owner does not need the vehicle fully charged, then it has a completely different kind of flexibility. The owner can pick an option in which, if the vehicle gets charged, the owner pays a small price, and if the vehicle does not charge or is charged partially, then the owner gets some compensation for unmet demand. Thus, the generator offers random generation to customers who can tolerate the lack of electricity at the time of delivery, wherein the generator guarantees no monetary losses in this transaction. In this paper, we study the associated optimization and pricing problem in such a market. We believe that utility scale storage, cloud servers, and certain non-critical load serving enetities can participate in this market. This paper is an expanded version of our previous work in \cite{dakhil2018cdc}; we refer to related results in \cite{dakhil2018cdc} whenever needed. Further, due to space constraints, we present the proofs of our results in archived version of this paper in \cite{dakhil2018arxiv}.

\subsection{Organization of the Paper}
We first formulate the seller's stochastic program in Section \ref{sec:pf}. In Section \ref{sec:alloc}, we analyze and solve the optimization problem for the efficient allocation and show that it satisfies a monotonically increasing property with respect to the strategic parameter. Section \ref{sec:price} reviews our pricing scheme, which exploits monotonicity in allocation, and studies consumers' surplus and seller's profit. Simulation results are presented in Section \ref{sec:simulation}.  Finally, we provide concluding remarks in section \ref{sec:con}. All proofs are available in the appendices.

\section{Problem Formulation}\label{sec:pf}
We have one generator who wants to sell its random generation to $N$ LSEs. Let $(\Omega,\ALP F,\mathbb{P})$ be a standard probability space. Let $w(\omega)$ be the electricity generated by the generator in the real time market. We use $f:(-\infty,\infty) \rightarrow[0,\infty)$ and $F:(-\infty,\infty) \rightarrow[0,1]$ to denote the density function and cumulative distribution function of the random variable $w(\omega)$, respectively. They are given by
\beqq{F(z) = \pr{w(\omega)\leq z} = \int_{-\infty}^z f(w) dw.}
We restrict $f(z) = 0$ for $z\leq 0$ and $f(z)>0$ for $z> 0$. Note that $F$ is a monotonically increasing differentiable function because it admits a density function with respect to a Lebesgue measure on the real line.

The generator asks LSEs to report their bids in the day ahead market. The $i$-th LSE responds by reporting a bid that consists of two parameters, $(c_i,\pi_i)$: $c_i$ represents the price per unit electricity LSE $i$ is willing to pay. This corresponds to a perfectly elastic demand represented by a constant inverse demand function. We make this assumption to simplify the analysis of allocating an infinitely divisible good, i.e. electricity, that has infinite dimensional bids. The parameter $\pi_i$ is the penalty per unit of electricity paid by the generator to LSE $i$ who accepts to receive this amount as a form of compensation in case of shortfall in the real time market. Let $x_i$ denote the amount of power contracted with LSE $i$ in the day ahead market, where $1\leq i \leq N$. We use $y_i(x;\omega)$ to denote the shortfall of LSE $i$ in the real time market, which would be a function of the allocation $x:=(x_1,\ldots,x_N)$ and randomness $\omega$. Define $y:=(y_1,...,y_N)$.

To simplify the analysis, we assume that there is no cost to dump any excess generation in real time. Therefore, the problem is only concerned with the case the realized generation is less than or equal to the sum of allocated electricity. Also, note that there are no network constraints in this formulation, where we assume that the transmission capacity is large enough to support the committed allocation.
 

\subsection{Key Questions for the Auction Design}

The generator designs a mechanism that determines: (i) the allocation $x^*$ of electricity in the day ahead market and the shortfall $y^*$ in the real time market, and (ii) the prices $p_i^*$ that each LSE $i$ needs to pay in exchange of $x_i^*$ units of electricity. Define $p^* = (p_1^*,\ldots,p_N^*)$. The tuple $(x^*,y^*,p^*)$ denotes the mechanism designed by the generator. Since the allocation and the prices are based on the bids, buyers have motive to lie about their bids $\{(c_i,\pi_i)\}_{i=1}^N$, if that improves their (expected) payoffs. In this case, it is desirable to devise the mechanism $(x^*,y^*,p^*)$ which satisfies some desirable properties stated below.


We study here the case in which LSEs strategically choose $c=(c_1,\ldots,c_N)$, whereas $\pi=(\pi_1,\ldots,\pi_N)$ is held constant and is known in advance. A typical example of this scenario is where electricity is procured from traditional generators in real time market to make up for the shortfall at known prices. The utility $u_i$ of LSE $i$ is a function of the allocation $x_i$ and the price $p_i$, and takes this quasilinear form
\beqq{u_i(x_i,p_i) = c_ix_i - p_i.}
The LSEs, as rational agents, act to maximize their individual utility functions, which, as the last equation shows, depend on the bids of all LSEs through the allocation and price. Let $s$ be the true value of all LSEs and denote the utility of LSE $i$ at $x_i^*$ and $p_i^*$ when it bids $\hat s_i$ and the others bid $\hat s_{-i}$ as
\beqq{U_i(\hat s_i,\hat s_{-i}) = c_ix_i^*(\hat s_i,\hat s_{-i}) - p_i^*(\hat s_i,\hat s_{-i}).}

On the other hand, the seller's expected utility at $(x^*,y^*,p^*)$ is given by
\beqq{U_0(\hat s)=\sum_{i=1}^Np_i^*(\hat s_i,\hat s_{-i})-\sum_{i=1}^N\ex{\pi_iy_i^*(x^*(\hat s_i,\hat s_{-i}),\omega)}.}
The goal of the generator is to devise the mechanism $(x^*,y^*,p^*)$ that is {\it efficient, incentive compatible in dominant strategy, individually rational, and budget balanced}. We formally define these properties below:


\begin{definition}\label{def:efficient}
A mechanism $(x^*,y^*,p^*)$ is efficient if the allocation $(x^*,y^*)$ maximizes social welfare $W(x,y) =  \sum_{i=1}^{N} \Big(c_ix_i(s) - \mathbb{E}\big[\pi_iy_i(x(s);\omega) \; \big]\Big)$.
\end{definition}

\begin{definition}\label{def:IC}
A mechanism $(x^*,y^*,p^*)$ is incentive compatible in dominant strategy if for every LSE $i\in\{1,\ldots,N\}$, if it reports $\hat s_i$ while its true bid is $s_i$, then
\beqq{U_i({s_i},\hat s_{-i})\geq U_i(\hat s_i,\hat s_{-i}),}
for all $s_i$, $\hat s_i$, and $\hat s_{-i}$.
\end{definition}
\begin{definition}\label{def:IR}
A mechanism $(x^*,y^*,p^*)$  is individually rational if for every LSE $i\in\{1,\ldots,N\}$,
\beqq{U_i({ s_i}, s_{-i})\geq 0\quad \text{ for all } s_i \text{ and } s_{-i}.}
\end{definition}
\begin{definition}\label{def:BB}
A mechanism $(x^*,y^*,p^*)$ is budget balanced if the generator's payoff satisfies $U_0(s)\geq 0$ for any bid vector $s$.
\end{definition}

\subsection{Solution Approach}
To devise an efficient mechanism, the generator formulates a stochastic program to compute an allocation that maximizes social welfare (see Definition \ref{def:efficient}). This translates to the following two-stage optimization problem with recourse:

\begin{equation}
\label{eqn:stp}
\begin{aligned}
 \underset{x\in\Re^N, y:\Re^N\times\Omega \rightarrow\Re^N}{\text{max}} & \; \Bigg\{ \sum_{i=1}^{N} c_ix_i - \mathbb{E}\Bigg[\sum_{i=1}^{N} \pi_iy_i(x;\omega) \; \Bigg]\Bigg\}\\
\text{subject to } & x_i \geq 0,  \;  \forall i \\
& 0\leq y_i(x;\omega) \leq x_i,  \;  \forall i, \\
&\sum_{i=1}^{N} y_i(x;\omega) =\Big(\sum_{i=1}^{N} x_i -w(\omega)\Big)^+.
\end{aligned}
\end{equation}
where the expectation is with respect to $w$. Since the utility of each LSE is affine in its bid parameter $c_i$, we use Myerson's payment scheme \cite{myerson1981} to determine the payment structure $p^*$. We show that for this problem, this payment scheme, coupled with the allocation computed by solving \eqref{eqn:stp}, induces a mechanism $(x^*,y^*,p^*)$ that is incentive compatibile in dominant strategies, individually rational, and budget balanced.

\section{Optimal Allocation and Shortfall Scheme}\label{sec:alloc}
The two-stage stochastic program formulated in \eqref{eqn:stp} is solved in our previous work \cite{dakhil2018cdc}. We review the main steps in deriving the allocation and shortfall scheme for completeness here. We solve (\ref{eqn:stp}) backwards in two steps. First we minimize the cost function of the second stage with respect to the decision variable $y$ subject to relevant constraints. Then, we use the optimal solution of the second stage to compute the optimal allocation among the competing LSEs. We first define some notation and make the following assumption on the bids of the LSEs.

Define $c_0 = \pi_0 = 0$. For $i\neq N$, let $\alpha_i$ and $\beta_i$ be defined as 
\beq{ \alpha_i &:= \frac{c_{i+1}(\pi_i - \pi_{i-1}) +c_{i-1}(\pi_{i+1} - \pi_i)}{\pi_{i+1} - \pi_{i-1}},\label{eqn:alpha}\\
\beta_i &:= c_{i+1}-(\pi_{i+1}-\pi_i).\nonumber\label{eqn:beta}}
Let $\alpha_N = c_{N-1}$ and $\beta_N = c_{N-1}+(\pi_N - \pi_{N-1})$.
\begin{assumption}\label{assm:bids}
\begin{enumerate}
    \item The penalty is ordered such that $\pi_1<\pi_2<...<\pi_N$.
    \item The value of LSE $i$ satisfies 
\beqq{\alpha_i <c_i <\beta_i.}
\end{enumerate}
\end{assumption}
%

If all agents have distinct penalty parameter, then Assumption \ref{assm:bids} (1) can be made without loss of generality. If Assumption \ref{assm:bids} (2) is violated by LSE $i$, then the generator will allocate either zero or infinite units of electricity to LSE $i$. We place this assumption on all the statements proved in the paper. The following theorem presents the solution to this stochastic program, the proof of which is outlined in \ref{sec:S2} and \ref{sec:S1}.

\begin{theorem}\label{thm:stp}
Consider the two-stage stochastic program in \eqref{eqn:stp} in which the bids satisfy Assumption \ref{assm:bids}. The optimal shortfall $y^*$, as a function of allocation $x$ and randomness $\omega$, is given by

\begin{equation}\label{eq:yistar}
 y^*(x;\omega)= \begin{cases} 
 \big(x_1,x_2,...,x_{k_x(\omega)-1},\varphi_{k_x(\omega)}(x) - w(\omega)\\
 0,0,...,0\big)  \qquad \text{if}~\textbf{1}_N \cdot x > w(\omega)\\
 \textbf{0} \qquad \qquad \qquad \text{else}
 \end{cases},
 \end{equation}
 where $k_x(\omega)$ is given by
\begin{equation} 
\label{eqn:kxomega}
k_x(\omega)=  \begin{cases} 
i & \varphi_{i+1}(x) \leq  w(\omega)< \varphi_{i}(x) \\
N & w(\omega)\leq x_N
\end{cases},
\end{equation}
and $\varphi_j(x)=\sum_{k=j}^Nx_k = x_j+\ldots+x_N$. The optimal allocation $x^*$ is given by
\begin{align}
x_i^* &= F^{-1}\left( \frac{c_i-c_{i-1}}{\pi_i - \pi_{i-1}} \right) - F^{-1}\left( \frac{c_{i+1}-c_i}{\pi_{i+1} - \pi_i} \right) \label{eqn:xistar}\\
x_N^* &= F^{-1}\left( \frac{c_N-c_{N-1}}{\pi_N-\pi_{N-1}} \right).\label{eqn:xNstar}\nonumber\\
\end{align}
\end{theorem}

An outline of the proof of the theorem is provided below. 

\subsection{Stage 2 Optimization Problem}\label{sec:S2}
Given a realization of wind generation,  $w(\omega)$, we first minimize the linear cost objective inside the expectation by solving the following linear program: 
\begin{equation} 
\label{eqn:Qomegax}
\begin{aligned}
Q(x;\omega) = & \;\underset{y}{\text{min}}
& &\pi \cdot y(x;\omega)\\
& \; \text{s.t.} & & 0 \leq y(x;\omega) \leq x, \\
& & & \textbf{1}_N \cdot y (x;\omega)= \Big( \textbf{1}_N \cdot x -w(\omega)\Big)^+,\\
\end{aligned}
\end{equation}
where $\textbf{1}_N$ is a column vector of $N$ ones. Recall that $\pi_1<\pi_2<...<\pi_N$. The optimal solution to this linear program, $y^*(x;\omega)$ is a random vector given by
\begin{equation*}
 y^*(x;\omega)= \begin{cases} 
 \big(x_1,x_2,...,x_{k_x(\omega)-1},\varphi_{k_x(\omega)}(x) - w(\omega),\\
 0,0,...,0\big)  \qquad \text{if}~\textbf{1}_N \cdot x > w(\omega)\\
 \textbf{0} \qquad \qquad \qquad \text{else}
 \end{cases}
 \end{equation*}
where $k_x(\omega)$ is given in \eqref{eqn:kxomega}. Note that the components of $y$ are zeros for all indices that are greater than $k_x(\omega)$. The rationale behind this solution is when wind is generated, it is more cost effective to allocate the available wind to the more expensive buyers and compensate the less expensive LSEs.

We use $y^*(x;\omega)$ to compute the optimal value of the objective function of the linear program in (\ref{eqn:Qomegax}). To simplify the computation we use the indicator function, $\mathds{1}_{\{\cdot\}}$,   which takes the value 1 if $\{\cdot\}$ is satisfied and 0 otherwise. Then, the optimal value can be written as
\begin{align}
\label{eqn:Qxw}
Q(x;\omega) & =\sum_{i=1}^{N}\pi_i y_i^*(x;\omega)\nonumber \\
& = \sum_{i=1}^N \pi_i \bigg[\mathds{1}_{\lbrace k_x(\omega)>i\rbrace} x_i+\mathds{1}_{\lbrace k_x(\omega)=i\rbrace} \varphi_i(x)\nonumber \\
&\qquad \qquad \qquad -\mathds{1}_{\lbrace k_x(\omega)=i\rbrace} w(\omega)\bigg],
\end{align}
where $\varphi_i$ is as defined in \eqref{eqn:kxomega}. We proceed to the next task, computing the expected value of the second stage, 
\beq{V(x):=\ex{Q(x;\omega)}.\nonumber}
We define the function $G:[0,\infty) \rightarrow \mathbb{R}$ to be  
\begin{equation}
G(z)=\int_0^z wf(w)dw. \label{eqn:Gz}
\end{equation}
and use it to evaluate $V(x)$.

\begin{lemma}
\label{lem:value}
Recall that $\varphi_j(x) = \sum_{k=j}^N x_k$. The second stage expected value function, $V(x) = \ex{Q(x;\omega)}$, is 
\begin{align}\label{eqn:value}
V(x) &=\sum_{i=1}^{N}\pi_i x_iF\Big(\varphi_{i+1}(x)\Big)+\sum_{i=1}^{N}\pi_i \varphi_{i}(x)F\Big(\varphi_{i}(x)\Big)\nonumber\\
&-\sum_{i=1}^{N}\pi_i \varphi_{i}(x)F\Big(\varphi_{i+1}(x)\Big)-\sum_{i=1}^{N}\pi_i G\Big(\varphi_{i}(x)\Big)\nonumber\\
&+\sum_{i=1}^{N}\pi_i G\Big(\varphi_{i+1}(x)\Big).
\end{align}
Further, $V$ is a deterministic and convex function of $x$.
\end{lemma}
\begin{proof}
See  Appendix \ref{app:value}.
\end{proof}

Having all the terms of the second stage expected value as functions of $x$, we finished the formulation of the cost function of the first stage and we can proceed to find the optimal, welfare maximizing, allocation $x^*$.
\subsection{Stage 1 Optimization Problem}\label{sec:S1}
To solve the first stage optimization problem, we first rewrite it as a minimization of a cost function. This cost function is composed of the expected cost of serving the shortfall in generation at the second stage minus the total value resulting from the allocation. This minimization problem is written as:
\begin{equation*}
\begin{aligned}
& \underset{x}{\text{min}} \; h(x) := -c\cdot x+V(x) \\
& \text{subject to} \; x\geq 0, x\in\Re^N.
\end{aligned}
\end{equation*}
We note that this is just a problem of minimizing a convex objective over a positive orthant. The derivative $\frac{\partial h}{\partial x_i}$ is computed as
\begin{equation}
\label{eqn:hxi}
\frac{\partial h}{\partial x_i}=-c_i+\sum_{j=1}^{i}(\pi_j-\pi_{j-1})F\Big(\varphi_{j}(x)\Big)\quad 1\leq i \leq N.
\end{equation}
The necessary and sufficient conditions for optimality given in Example 2.1.1 from \cite{bertsekas1999} state that at the optimal allocation, $x^*$, the derivative of $h(x)$ is nonnegative, that is $\frac{\partial h}{\partial x_i}|_{x^*}\geq 0$. This implies that
\beq{\label{eqn:iff}x^* \text{ is optimal} \Longleftrightarrow c_i\leq \sum_{j=1}^{i}(\pi_j-\pi_{j-1})F\Big(\varphi_{j}(x^*)\Big).}
Furthermore, 
\begin{enumerate}
\item if $x_i^*>0$, then  $c_i = \sum_{j=1}^{i}(\pi_j-\pi_{j-1})F\Big(\varphi_{j}(x^*)\Big)$;
\item if $c_i < \sum_{j=1}^{i}(\pi_j-\pi_{j-1})F\Big(\varphi_{j}(x^*)\Big)$, then $x_i^*=0$. 
\end{enumerate}

If $x^*$ is strictly positive, the following equations are satisfied
\beq{x_i^*+\ldots+x_N^* = F^{-1}\left( \frac{c_i-c_{i-1}}{\pi_i-\pi_{i-1}} \right),\quad 1\leq i\leq N.\label{eqn:optcond}}
By solving the equations above, we arrive at the formula of $x^*$ given in (\ref{eqn:xistar})-(\ref{eqn:xNstar}). The next Lemma provides conditions for LSE to be allocated zero or positive amount of electricity.

\begin{lemma}\label{lem:ci}
Suppose that Assumption \ref{assm:bids} holds. Then, the allocation $x_i^*$ is strictly positive for all $i$.
\end{lemma}
\begin{proof}
Refer to Appendix \ref{app:ci}.
\end{proof}

This proves Theorem \ref{thm:stp}. We now turn our attention to the pricing scheme in the next section.

\section{The Pricing Scheme}\label{sec:price}
Recall that the LSEs bid their inverse demand functions to the renewable generator and the penalty for the shortfall is known. In such a situation, the LSEs have inherent incentives to misrepresent their true costs or penalty to maximize their profit. Thus, LSEs can be strategic in deciding on their bids. Under such a situation, it is desired to devise allocation and pricing policies jointly that can induce the LSEs to bid their true parameters. Here we show that the allocation policy determined in Theorem \ref{thm:stp}, together with a pricing policy based on Myerson's lemma \cite{myerson1981}, \cite{roughgarden2016}, induces truthful behavior in dominant strategy with respect to the willingness to pay parameter. 

To simplify the analysis, we first study the problem in which $\pi_1,\ldots,\pi_N$ are the non-strategic components of the LSEs' bids, and $c_1,\ldots,c_N$ are strategically chosen by the bidders to maximize their utilities. The motivation to study this problem is that $c_i$ captures the value generated by consumption of one unit of electricity by LSE $i$, $\pi_i$ would represent the amount LSE $i$ would pay to procure the shortfall from a traditional gas-fired or coal based generator to meet the demand.

For next result, let us write the optimal allocation $x^*$ as a function of the bids $c$. For clarity, let  $c_{-i} = (c_1,...,c_{i-1} ,c_{i+1},...,c_N)$, and we write $x_i^*(s,c_{-i})$ to denote the optimal allocation when LSE $i$ bids $s$ and all other LSEs bid $c_{-i}$. In the next lemma, we show that as an LSE increases its bid from $0$ to $\beta_i$, its optimal allocation is monotonically increasing. 
\begin{lemma}
\label{lem:mono}
For fixed $c_{-i}$, the optimal allocation $x^*_i(s,c_{-i})$ is monotonically increasing in $s \in (0,\beta_i)$ for all $1\leq i\leq N$.	
\end{lemma}
\begin{proof}
\mytest{The proof is given in \cite{dakhil2018cdc}.}{See Appendix \ref{app:mono}.}
\end{proof}

Myerson's Lemma \cite{myerson1981} states that for a monotonic allocation function, the payment rule given by the expression
\begin{align*}
p_i^*(c_i,c_{-i})=\int_{0}^{c_i}s \frac{d}{ds}x_i^*(s,c_{-i})ds, \quad 1\leq i\leq N
\end{align*}
leads to incentive compatibility in dominant strategies in a single-parameter auction where valuation is linear in the bidding parameter (which is the case here). We evaluate this integral in the following theorem.

\begin{theorem}
\label{thm:DSIC}
Given the allocation function $x^*(c_i,c_{-i})$, define the payment scheme $p_i^*$ as
\beq{p_i^*(c_i,c_{-i}) & = c_ix_i^*(c_i,c_{-i})\nonumber\\
&+(\pi_{i+1}-\pi_{i-1})G\bigg(F^{-1}\bigg(\frac{c_{i+1}-c_{i-1}}{\pi_{i+1}-\pi_{i-1}}\bigg)\bigg)\nonumber\\
&-(\pi_{i+1}-\pi_{i})G\bigg(F^{-1}\bigg(\frac{c_{i+1}-c_i}{\pi_{i+1}-\pi_{i}}\bigg)\bigg)\nonumber\\
 &-(\pi_{i}-\pi_{i-1})G\bigg(F^{-1}\bigg(\frac{c_{i}-c_{i-1}}{\pi_{i}-\pi_{i-1}}\bigg)\bigg),\label{eqn:pi}}
\beq{p_N^*(c_N,c_{-N}) &=c_Nx_N^*(c_N,c_{-N})\nonumber\\
&-(\pi_{N}-\pi_{N-1})G\bigg(F^{-1}\bigg(\frac{c_{N}-c_{N-1}}{\pi_{N}-\pi_{N-1}}\bigg)\bigg),\label{eqn:pN}}
where $G(\cdot)$ is defined in \eqref{eqn:Gz} and $1\leq i\leq N-1$. Then, the joint allocation payment scheme $(x^*,p^*)$ induces each LSE to bid truthfully in dominant strategies.
\end{theorem}
\begin{proof}
The proof comprises evaluating the integral to compute the payment in accordance with Myerson's lemma and showing that no matter what other LSEs bid, the utility of LSE $i$ is maximized if it bids $c_i$.  
\mytest{A detailed proof is presented in \cite{dakhil2018cdc}.}{We refer to Appendix \ref{app:DSIC} for a detailed proof.}
\end{proof}

As one can observe in the expression for the payment scheme, each LSE pays the cost according to the product of its allocation and its reported valuation, but it receives a discount that is precisely based on bids of the preceding LSE and the succeeding LSE. 
\subsection{Individual Rationality of the Mechanism}
In this subsection, we prove that the LSEs are paying less than their valuation, which implies voluntary participation of players in the market -- this is the \textit{Individual Rationality} (IR) property in Definition \ref{def:IR}.  
\begin{lemma}\label{lem:plv}
For all $i$, the price player $i$ pays for electricity is not more than its valuation. that is,
\beqq{p_i^*(c_i,c_{-i})\leq c_ix_i^*(c_i,c_{-i}).}
\end{lemma}
\begin{proof}
The proof is a direct consequence of the fact that $G\circ F^{-1} :[0,1)\rightarrow[0,\infty)$ is a convex monotonically increasing function.  \mytest{We present the complete proof in \cite{dakhil2018arxiv}.}{We present the complete proof in Appendix \ref{app:plv}.}
\end{proof}
The following theorem proves that the mechanism $(x^*,y^*,p^*)$ is individually rational from the perspective of buyers.
\begin{theorem}\label{thm:IR}
The mechanism $(x^*,y^*, p^*)$ is individually rational, that is $U_i(c_i,c_{-i})\geq 0$ for all LSE $i$.
\end{theorem}
\begin{proof}
From lemma \ref{lem:plv}, we conclude that $U_i(c_i,c_{-i}) = c_ix_i^*(c_i,c_{-i}) -  p_i^*(c_i,c_{-i}) \geq 0$. Thus, the mechanism is individually rational.
\end{proof}

We just showed that flexible buyers lose nothing when they buy the random generation through the defined auction. 

Since the generator is taking the financial responsibility of compensating any possible shortfall in generation, it is possible for generators to make losses with some probability. Expected profit of the generator measures the long term average of losses and profits the generator expects to make.Next, we  study the expected payoff of the generator.  

\subsection{Budget Balancedness of the Mechanism}
Recall Definition \ref{def:BB}: A mecahnism is budget balanced if there is no net transfer from the mechanism to the players. In our case, this translates to the generator ending with nonnegative utility. The generator's expected payoff is
\beqq{U_0^*:=U_0^*(c)=\sum_{i=1}^N\big(p^*_i-\ex{\pi_iy_i^*(x^*;\omega)}\big).}
In the next theorem,  we present a lower bound on the profit of the generator.
\begin{theorem}\label{thm:bound}
Suppose that $F(x)$ is convex over $(0,x^*_N)$. Then, the generator's expected profit is lower bounded by
\beqq{&U_0^*\geq\sum_{i=1}^{N-1}\Bigg[ \big(c_i-\pi_iF(\varphi_i(x^*)\big)x_i^*(c_i,c_{-i})+\frac{\pi_{i-1}(c_i-\alpha_i)}{\pi_i-\pi_{i-1}}\\
&\times\bigg[F^{-1}\Big(\frac{c_{i+1}-c_{i-1}}{\pi_{i+1}-\pi_{i-1}}\Big)-F^{-1}\Big(\frac{c_{i+1}-c_i}{\pi_{i+1}-\pi_i}\Big)\bigg] \Bigg]\\
&+\Big[\frac{c_{N-1}\pi_N}{\pi_N-\pi_{N-1}}-\frac{c_N+c_{N-1}}{2}\frac{\pi_{N-1}}{\pi_N-\pi_{N-1}}\Big]x_N^*(c_N,c_{-N}).}
\end{theorem}
\vspace{0.5em}
\begin{proof}
We break the process of bounding $U_0^*$ by lower bounding $p_i^*-\pi_i^*\ex{y_i^*(x^*;\omega)}$ for every $1\leq i\leq N$ in Lemma \ref{lem:bound}.
\end{proof}

\begin{lemma}\label{lem:bound}
Under the mechanism $(x^*,y^*,p^*)$, the generator's expected payoff from each LSE $i\in\{1,\ldots,N-1\}$ is lower bounded by
\small{
\beqq{& p_i^*(c_i,c_{-i}) - \ex{\pi_iy_i^*(x^*;\omega)}\geq \big(c_i-\pi_iF(\varphi_i(x^*)\big)x_i^*(c_i,c_{-i})\\
& +\frac{\pi_{i-1}(c_i-\alpha_i)}{\pi_i-\pi_{i-1}}\bigg[F^{-1}\Big(\frac{c_{i+1}-c_{i-1}}{\pi_{i+1}-\pi_{i-1}}\Big)-F^{-1}\Big(\frac{c_{i+1}-c_i}{\pi_{i+1}-\pi_i}\Big)\bigg].}}
Moreover, if $F(x)$ is convex over $(0,x^*_N)$, then
\beqq{&p_N^*(c_N,c_{-N}) - \ex{\pi_N y_N^*(x^*;\omega)}\geq\\
&\Big[\frac{\pi_Nc_{N-1}}{\pi_N-\pi_{N-1}}-\frac{c_N+c_{N-1}}{2}\frac{\pi_{N-1}}{\pi_N-\pi_{N-1}}\Big]x_N^*(c_N,c_{-N}).}
\end{lemma}
\vspace{0.5em}
\begin{proof}
The proof simply follows from some algebraic manipulations. First, the expected value of the optimal recourse variable, $y_i^*(x;\omega)$ is computed to be
\beq{\ex{y_i^*(x;\omega)}&=x_iF\big(\varphi_{i+1}(x)\big)+\varphi_{i}(x)F\big(\varphi_{i}(x)\big)\nonumber\\
&-\varphi_{i}(x)F\big(\varphi_{i+1}(x)\big)-G\big(\varphi_{i}(x)\big)\nonumber\\
&+G\big(\varphi_{i+1}(x)\big).\label{eq:Ey1}}
Using \eqref{eq:Ey1} and the payments from Theorem \ref{thm:DSIC}, and by exploiting the convexity of $G\circ F^{-1}$, we arrive at the lower bound. 
\mytest{We refer the reader to  \cite{dakhil2018arxiv} for a detailed proof.}{See Appendix \ref{app:bound} for the proof.}
\end{proof}
Therefore, the total expected payoff, i.e. the generator's expected profit, is bounded by the sum of the bounds in Lemma \ref{lem:bound}. Under certain additional assumptions on $(c_i, \pi_i)$, we can show that the expected profit is strictly positive.

\begin{lemma}\label{lem:pp}
If $c_i/\pi_i < c_{i-1}/\pi_{i-1}$ for all $2\leq i\leq N$, then $U_0^*> 0$. Consequently, the mechanism $(x^*,y^*,p^*)$ is budget balanced.
\end{lemma}
\begin{proof}
Refer to Appendix \ref{app:pp}. 
\end{proof}

We evaluate the expected profit and the lower bound for the case where generation is distributed according to Weibull distribution in the next section. We note that our mechanism is budget balanced in expectation as we are concerned here with the expected profit. We also note that it is possible to accrue loss from some runs of the auction, though the long term average (of independent auctions) will always be positive.

\section{Numerical Simulation} \label{sec:simulation}
In this section, we simulate an instance of the auction where the bids are held fixed and the generation is assumed random. First, consider $\eta\in(0,1)$ and let
\beqq{c_i = \frac{1-\eta^i}{1-\eta} c_1,\quad  \pi_i = i\pi_1\quad \text{ for all } i\in\{1,\ldots,N\}.}
Let us define $\hat c = \frac{c_1}{\pi_1}<1$. With this definition, we get
\beqq{\frac{c_{i+1} - c_i}{\pi_{i+1} - \pi_i} = \frac{\eta^{i+1} - \eta^i}{1-\eta}\frac{c_1}{\pi_1} =  \eta^i \hat c. }
We note that  $(c,\pi)$ thus defined satisfy Assumption \ref{assm:bids}. The optimal allocation for LSE $i$, where $1 \leq i \leq N$, is
\beqq{x_i^* =  F^{-1}\Big(\eta^{i-1} \hat c\Big) - F^{-1}\Big(\eta^i \hat c\Big),\quad x_N^* =  F^{-1}\Big(\eta^{N-1} \hat c\Big).}
As for the source of renewable energy, we assume an aggregated wind generated electricity. Bradbury showed in \cite{bradbury2013} that such generation follows a Weibull distribution, which is characterized by a probability density function $f(w)$ of the form
\beqq{f(w) = \frac{k}{\lambda}\left( \frac{w}{\lambda} \right)^{k-1}\exp\left( - \left( \frac{w}{\lambda}\right)^k \right) \text{ for } w\geq 0,}
where $k$ is the shape parameter and $\lambda$ is the scale parameter, which is proportional to the mean power output. The corresponding cumulative distribution function is
\beqq{F(w) = 1-\exp\left( -\left( \frac{w}{\lambda} \right)^k \right) \text{ for } w\geq 0.}
In the same paper, it was shown that $k$ increases from $1.7$ to $3$ as the power output of more wind turbines are accumulated. For ease of exposition, we adopt $k=2$ in our computations. The corresponding $\lambda$ is found to be $1509$ kW, when the mean power is $1337$ kW.

In order to calculate the components $(x^*,p^*)$ of our DSIC mechanism , we need to evaluate both $F^{-1}(\rho)$ and $G(z)$ for the Weibull distribution. We note that $F(w)$ is a monotonically increasing function, therefore, it has an inverse function, which can be easily shown to be,
\beqq{F^{-1}(\rho)= \lambda\sqrt[k]{\ln\left( \frac{1}{1-\rho} \right)}= 1509\sqrt{\ln\left( \frac{1}{1-\rho} \right)}~ \text{kW}.}
In our case, $k=2$, computing $G(z)$ involves the evaluation of the following integral,
\beqq{G(z)&=\frac{2}{\lambda^2}\int_0^z  w^2e^{ - \left( \frac{w}{\lambda}\right)^2} dw=\lambda\int_0^{\frac{z^2}{\lambda^2}}\nu^{\frac{3}{2}-1}e^{-\nu}d\nu\\
&=\lambda\gamma\left(\frac{3}{2},\frac{z^2}{\lambda^2}\right),}
where we used the substitution $\nu=\frac{w^2}{\lambda^2}$, and $\gamma$ is the lower incomplete gamma function. This yields
\beqq{G\circ F^{-1} (\rho) = \lambda\gamma\left(\frac{3}{2},\frac{1}{1-\rho}\right).}

We simulate the case $N=5$ for different values of $\eta$ while fixing $\hat c=10/12$. Figure \ref{fig:alloc} and Figure \ref{fig:pay} illustrate how each buyer's allocation and payment change in response.

\begin{figure}[!ht]
\centering
  \includegraphics[width=0.8\linewidth]{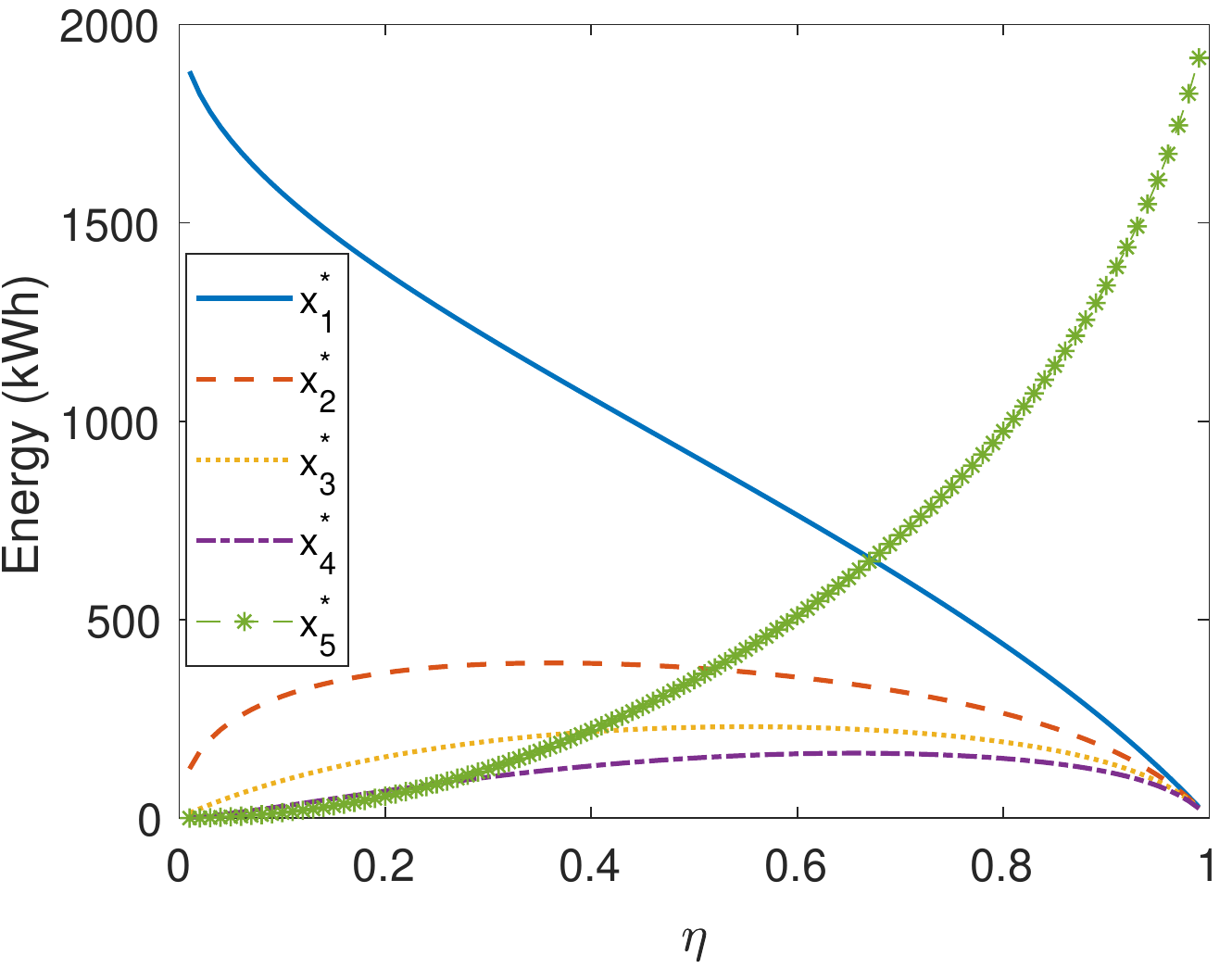}
  \caption{Allocation as a function of $\eta.$}
  \label{fig:alloc}
\end{figure}
\begin{figure}[!ht]
  \centering
  \includegraphics[width=0.8\linewidth]{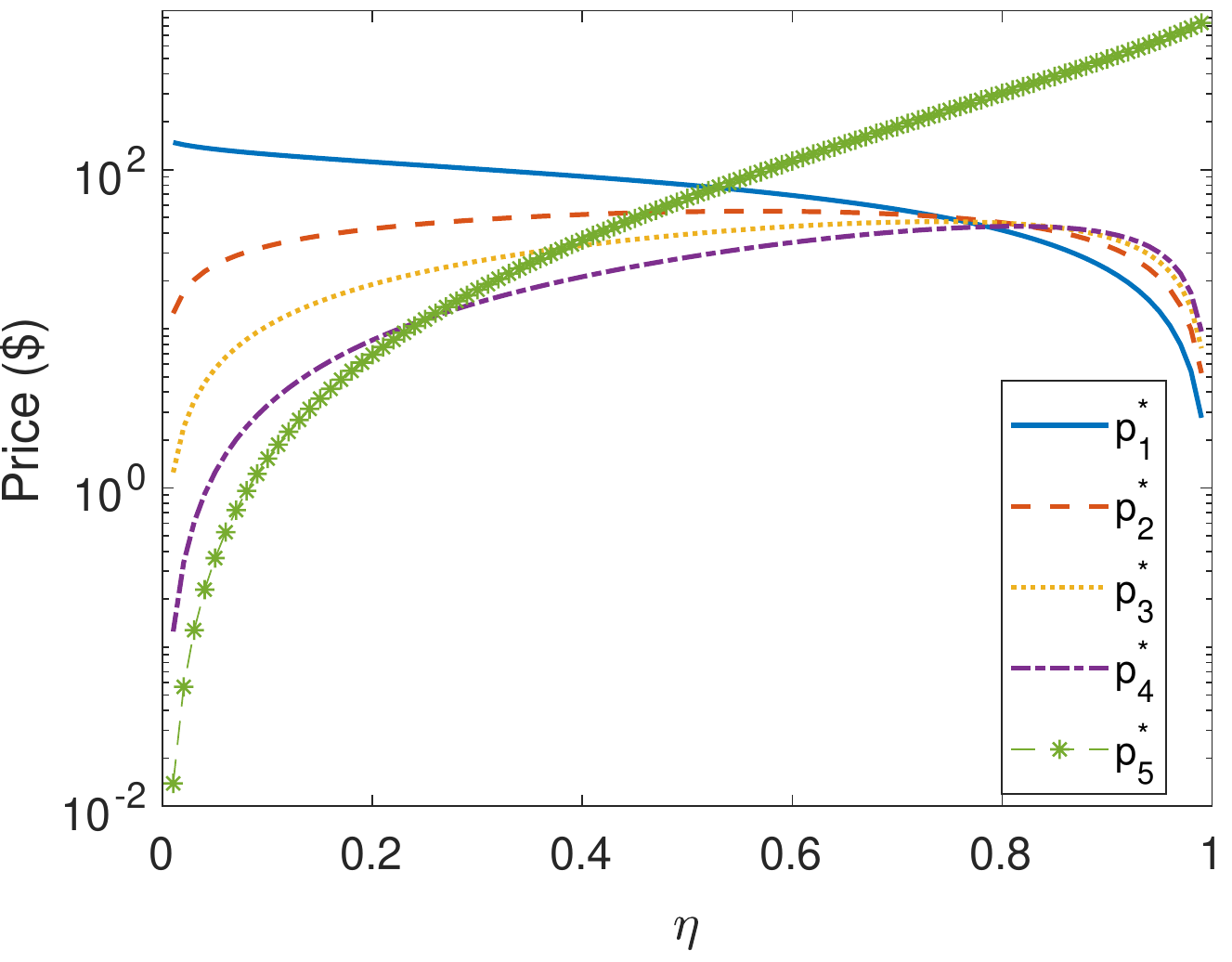}
  \caption{Myerson's Payments as a function of $\eta.$}
  \label{fig:pay}
\end{figure}

We note the contrary relationship between $x_1^*~ \text{and}~x_5^*$. When $\eta$ is low, the components of $c$ are close to each other, so more power is allocated to player 1, the cheapest player, and since most of the generation is allocated to him, not much is left to the other players. On the other hand, as $\eta$ increases, the components of $c$ spread out and the allocation assigns more energy to the LSEs with higher willingness to pay, in this case LSE 5. Yet, LSE 5, who has the highest valuation for the energy, has to pay far more than what player 1 has to pay per unit kWh. This higher payment makes up for the high expected compensation due to his high penalty $\pi_5$. This trend is clear in Figure \ref{fig:perunit} that shows the per unit prices of energy  for all the LSEs.
\begin{figure}[!ht]
\centering
  \includegraphics[width=0.8\linewidth]{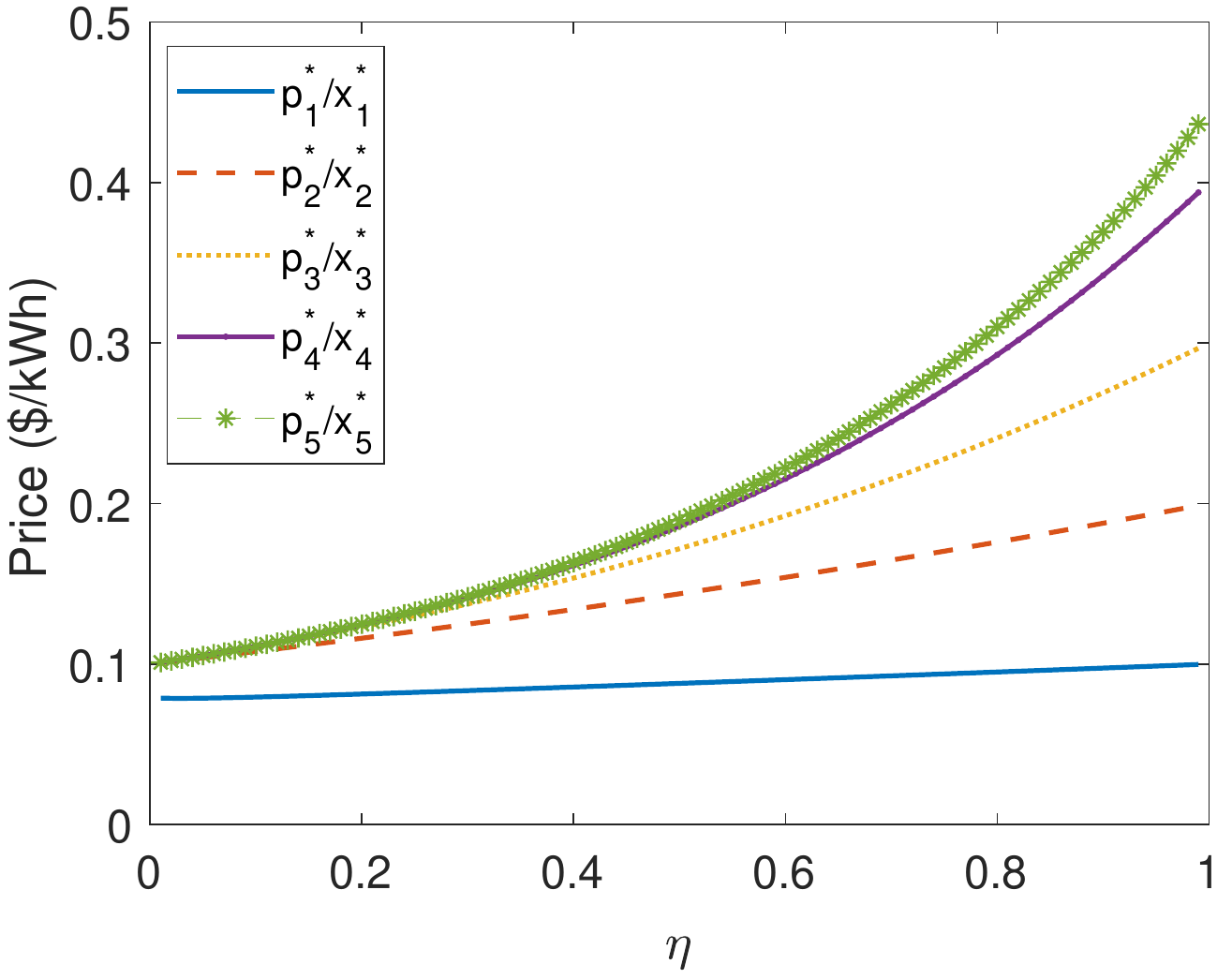}
  \caption{Prices per unit energy as a function of $\eta.$}
  \label{fig:perunit}
\end{figure}

It is worth mentioning that despite this non-uniform pricing, players are receiving a discount for their flexibility and their utilities are at their maximum by the DSIC feature of our proposed mechanism. We define the discount as 
\beqq{\text{discount} = 100\Big(\frac{c_ix_i^*-p_i^*}{c_ix_i^*}\Big)\%,
}The discounts are illustrated in Figure \ref{fig:discount}, where we see that even the expensive player is motivated by a decent amount of savings. Figure \ref{fig:surplus2} depicts the positive surplus of the players in this mechanism.
\begin{figure}[!ht]
\centering
  \includegraphics[width=0.8\linewidth]{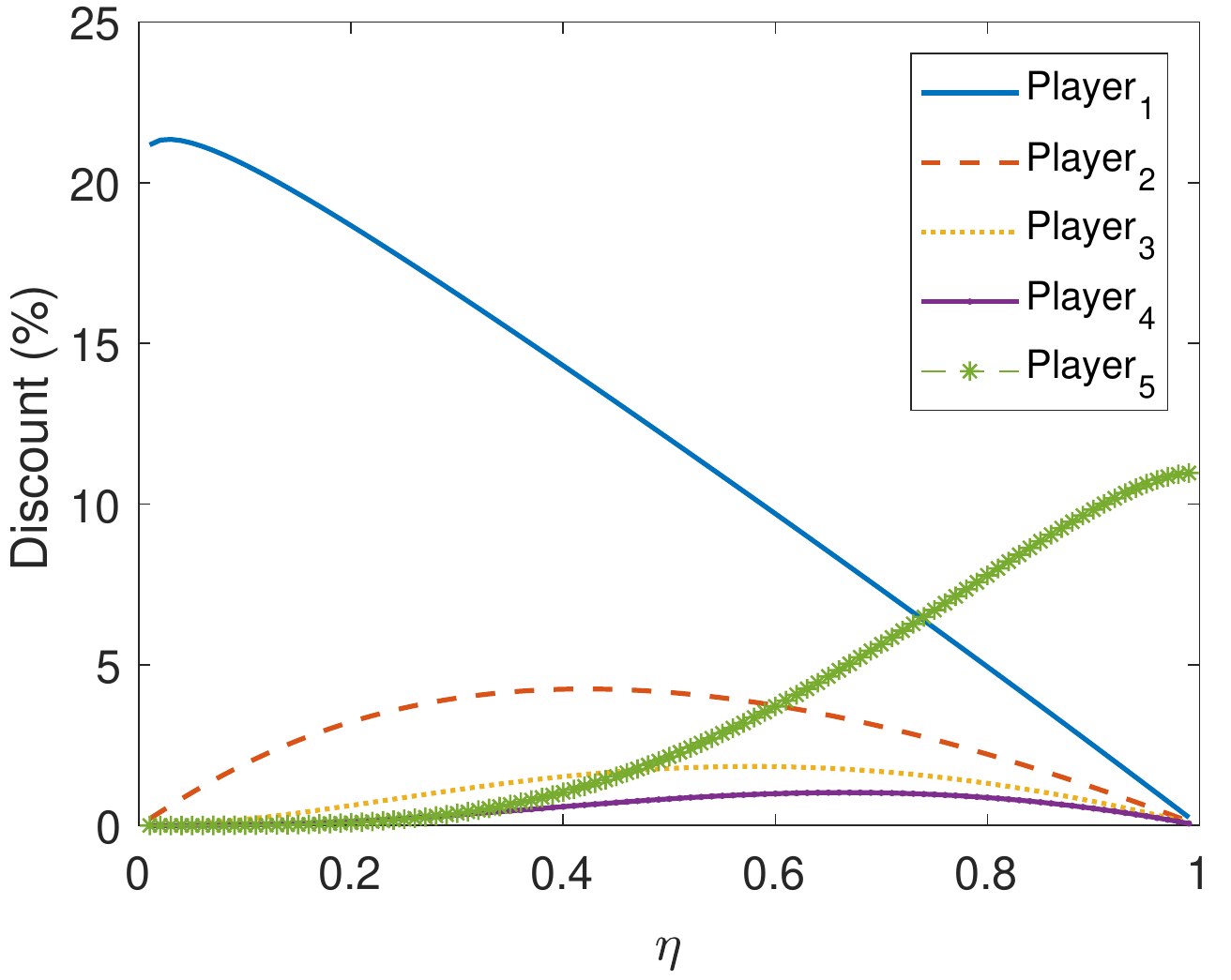}
  \caption{Discount as a function of $\eta$.}
  \label{fig:discount}
\end{figure}

\begin{figure}[!ht]
\centering
  \includegraphics[width=0.8\linewidth]{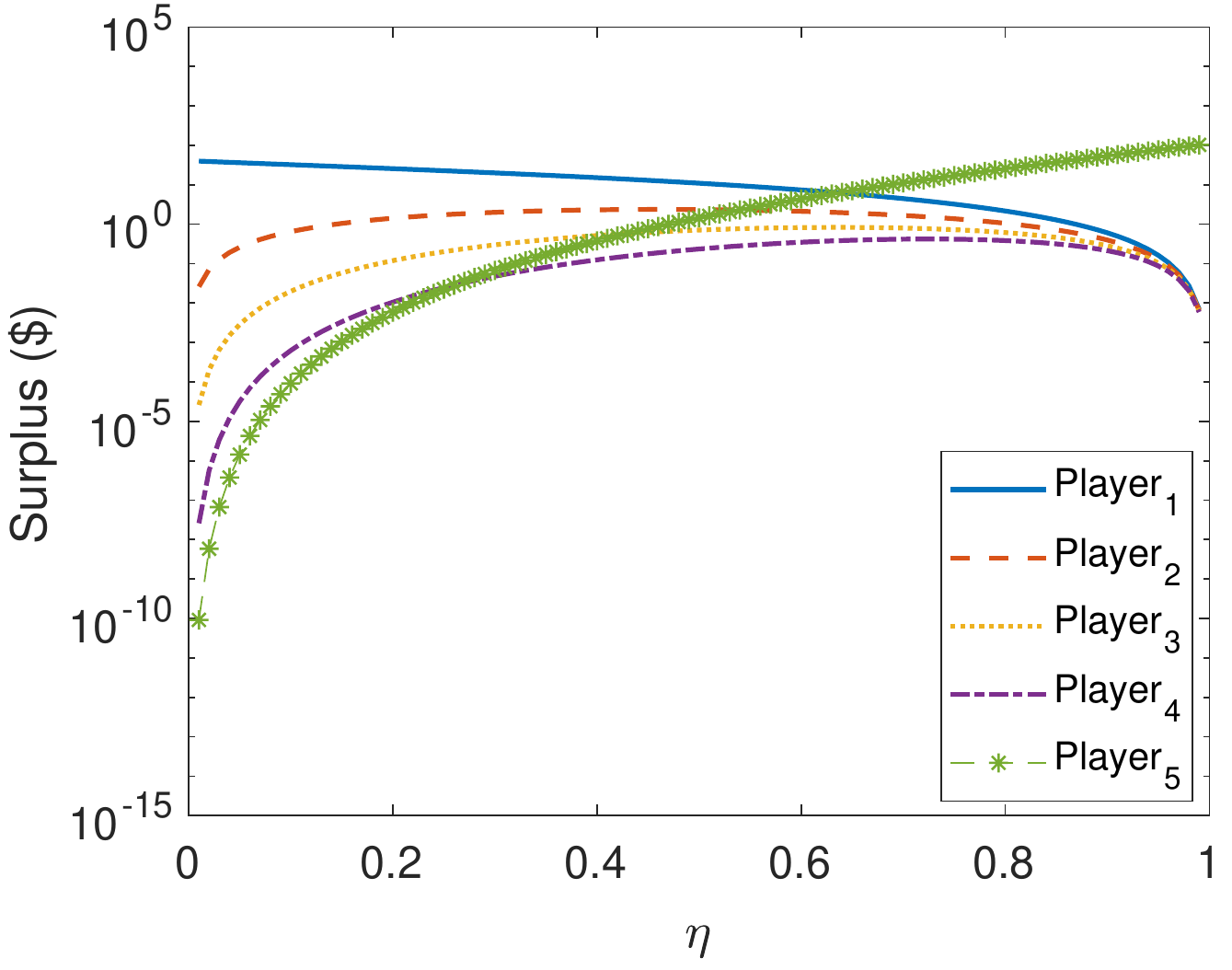}
  \caption{LSEs' utility as a function of $\eta$.}
  \label{fig:surplus2}
\end{figure}

As for the seller, the computation of the expected profit requires simulating an instance of the auction many times to find the expected compensation over independent runs of this auction. We observe that budget balancedness is achieved in expectation in Figure \ref{fig:BP}, where we compute the sample averaged expected profit for different values of $\eta$ when $\hat c =10/12$ along with the corresponding lower bound.
\begin{figure}[!ht]
\centering
  \includegraphics[width=0.8\linewidth]{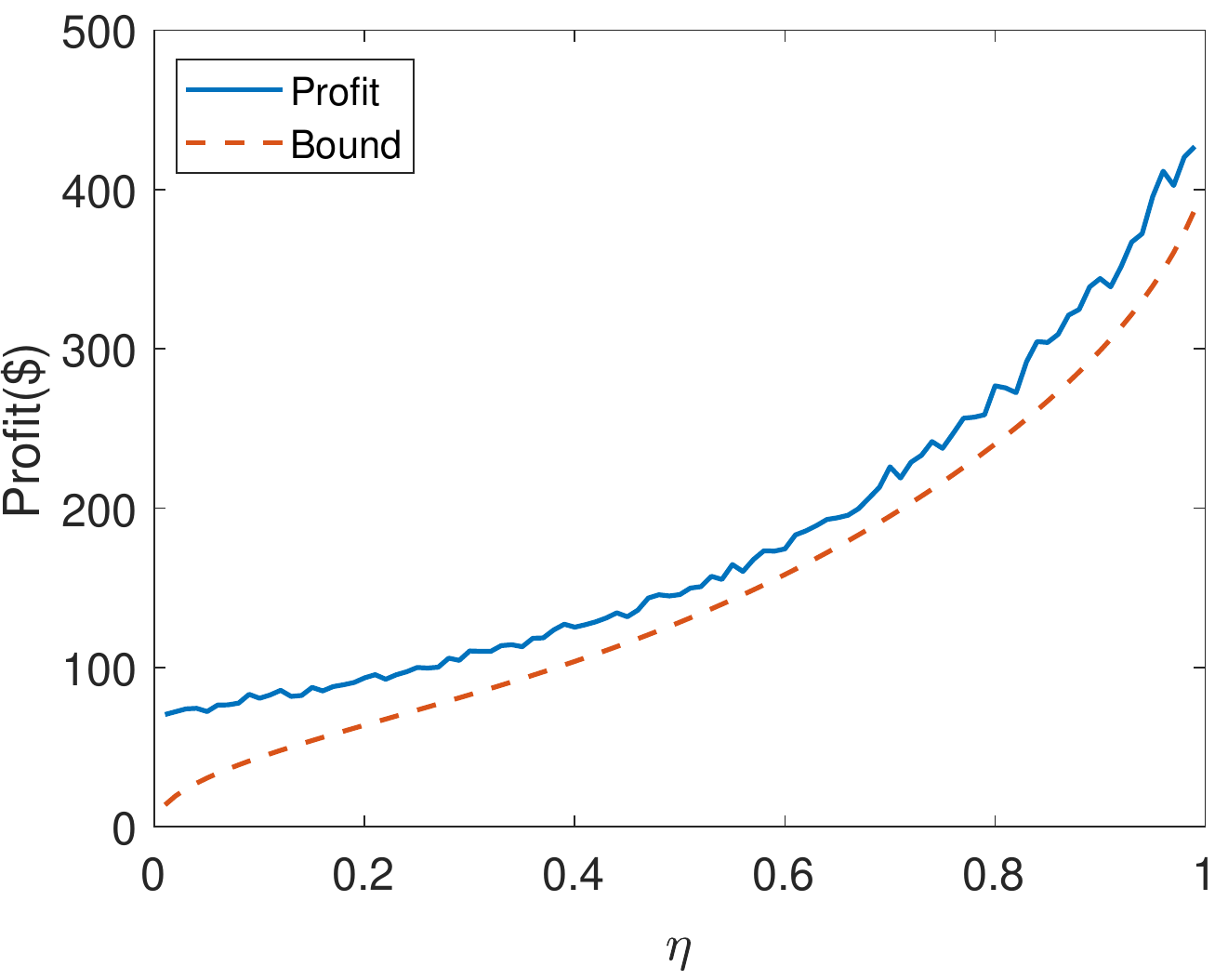}
  \caption{Generator's empirical average profit and lower bound on the profit for $\hat c =10/12$ as a function of $\eta$.}
  \label{fig:BP}
\end{figure}

\section{Conclusion}\label{sec:con}
Stochastic programming provides the generator with a powerful tool to make reliable decisions regarding how much electricity to contract for in the day ahead market even before knowing the actual generation in the real time market. There are many ways through which renewable electricity could be priced in this case. Motivated by considerations in storage and EV charging markets, we assumed a model in which some loads are open to the possibility of not consuming electricity at all if they are compensated appropriately. Using valuation bids from the LSEs, we devise an allocation and a corresponding payment scheme that yields not only truthful bidding in dominant strategies, but satisfies efficiency, individual rationality, and budget balancedness. Thus, we proved that LSEs would voluntarily participate in the proposed auction, and computed a lower bound on the profit of the generator. Further, we conducted numerical simulations under certain parametrized bids and with the assumption that the renewable generation is distributed according to Weibull distribution. The simulation results indicate that the consumer surplus is positive (although, for some LSEs, they are negligible), and the LSEs receive non-trivial discounts for their consumption.

\appendices
\section{Proof of Lemma \ref{lem:value}}\label{app:value}
By substituting the expression in \eqref{eqn:Qxw} in $V(x)=\mathbb{E}_{\omega}[Q(x;w(\omega))]$ and using the fact that the expectation of the indicator function is just the probability of the condition being true, we get
\begin{align}
\label{eq4}
V(x)&=\mathbb{E}_{\omega}[Q(x;w(\omega))] = \sum_{i=1}^N \pi_i \bigg [\mathbb{P} \lbrace k_x(\omega)>i\rbrace x_i  \nonumber \\
&+\mathbb{P}\lbrace k_x(\omega)=i\rbrace \varphi_i(x)\nonumber \\
& -\mathbb{E}_{\omega}[w(\omega)|k_x(\omega)=i]\mathbb{P}\lbrace k_x(\omega)=i\rbrace \bigg].
\end{align}
To compute the probabilities in \eqref{eq4}, we assume that wind generation has a smooth and continuous cumulative distribution function $F(w)=\mathbb{P}\lbrace w(\omega)\leq w\rbrace$ that is everywhere differentiable. Using the definition of $k_x(\omega)$ in \eqref{eqn:kxomega}, we get
\begin{align}
\label{eq5}
\mathbb{P}\lbrace k_x(\omega)=i\rbrace &= \mathbb{P}\bigg\lbrace \varphi_{i+1}(x)\leq  w(\omega)< \varphi_{i}(x) \bigg\rbrace \nonumber\\
&= F\bigg(\varphi_{i}(x)\bigg)-F\bigg(\varphi_{i+1}(x)\bigg),
\end{align}
for $1\leq i \leq N$, where we use the convention that a summation with lower limit greater than the upper limit is zero and $F(0)=0$.
Also,
\begin{align}
\label{eq6}
\mathbb{P}\lbrace k_x(\omega)>i\rbrace &=\sum_{j=i+1}^N\mathbb{P}\lbrace k_x(\omega)=j\rbrace \nonumber \\
&=\sum_{j=i+1}^N \bigg( F\big(\varphi_{j}(x)\big)-F\big(\varphi_{j+1}(x) \big)\bigg) \nonumber\\
&= F(\varphi_{i+1}(x)),
\end{align}
where all except the first term cancel each other out. Note that $\mathbb{P}\lbrace k_x(\omega)>N\rbrace=0$. 
To evaluate the last probabilistic term in \eqref{eq4}, we use the function $G$, defined in \eqref{eqn:Gz}. Now we obtain the desired concise expression
\begin{align}
\label{eq7}
& \mathbb{E}_{\omega}[w(\omega)|k_x(\omega)=i]\mathbb{P}\lbrace k_x(\omega)=i\rbrace\nonumber\\
&=\int_{\varphi_{i+1}(x)}^{\varphi_{i}(x)}w(\omega)f(\omega)d\omega =G\big(\varphi_{i}(x)\big)-G\big(\varphi_{i+1}(x)\big),
\end{align}
for $1\leq i \leq N$. The expression in the lemma follows immediately when we substitute \eqref{eq5}, \eqref{eq6}, and \eqref{eq7} in \eqref{eq4}.

We now prove that $V(x)$ is convex. Recall that $Q(x;\omega) = \pi^T y^*(x;\omega)$. By evaluating this expression (given in Theorem \ref{thm:stp}, we conclude that $y^*$ is linear in $w(\omega)$. 

We now show the convexity of $Q(\cdot,\omega)$ for a realization of the wind generation $w(\omega)$. Pick $x_1,x_2\in\Re^N$ such that $x_1,x_2\geq 0$. Define $y_1^*:=y_1^*(x_1;\omega)$ and $y_2^*:=y_2^*(x_2;\omega)$ as the optimal solutions associated with $Q(x_1;\omega)$ and $Q(x_2;\omega)$, respectively, then the vector $y=\alpha y_1^*+(1-\alpha)y_2^*$, where $\alpha \in [0,1]$, turns out to be a feasible vector that satisfies the constraints of the optimization problem in \eqref{eqn:stp} when $x=\alpha x_1+(1-\alpha) x_2$. Therefore, the value of the cost function at $y$ is greater than or equal to the optimal value $Q(x;\omega)$ and the following equation proves convexity,
\begin{equation*}
\begin{aligned}
Q(x;\omega)&=Q(\alpha x_1+(1-\alpha) x_2; \omega)\\
&\leq \pi^Ty\\
&=\pi^T(\alpha y_1^*+(1-\alpha)y_2^*)\\
&=\alpha\pi^T y_1^*+(1-\alpha)\pi^Ty_2^*\\
&=\alpha Q(x_1;\omega)+(1-\alpha)Q(x_2;\omega).
\end{aligned}
\end{equation*}
Indeed, the convexity of $\mathbb{E}_{\omega}[Q(x;\omega)]$ follows immediately by the linearity of the expectation operator, which preserves convexity.

\begin{figure*}[!t]
\begin{align*}
p_i^*(c_i,c_{-i})&=\int_{0}^{c_i}s \frac{d}{ds}x_i^*(s,c_{-i})ds=\int_{\alpha_i}^{c_i}s \frac{d}{ds}x_i^*(s,c_{-i})ds\\
&=\frac{1}{(\pi_i-\pi_{i-1})}\int_{\alpha_i}^{c_i}\frac{s}{f\bigg(\frac{s-c_{i-1}}{\pi_i-\pi_{i-1}}\bigg)}ds+\frac{1}{(\pi_{i+1}-\pi_{i})}\int_{\alpha_i}^{c_i}\frac{s}{f\bigg(\frac{c_{i+1}-s}{\pi_{i+1}-\pi_{i}}\bigg)}ds\\
&=\frac{1}{(\pi_i-\pi_{i-1})}\bigg[(\pi_i-\pi_{i-1})sF^{-1}\bigg(\frac{s-c_{i-1}}{\pi_i-\pi_{i-1}}\bigg)-(\pi_i-\pi_{i-1})\int F^{-1}\bigg(\frac{s-c_{i-1}}{\pi_i-\pi_{i-1}}\bigg)ds \bigg]\bigg\rvert_{\alpha_i}^{c_i}\\
&\qquad\qquad -\frac{1}{(\pi_{i+1}-\pi_{i})}\bigg[(\pi_{i+1}-\pi_{i})sF^{-1}\bigg(\frac{c_{i+1}-s}{\pi_{i+1}-\pi_{i}}\bigg)-(\pi_{i+1}-\pi_{i})\int F^{-1}\bigg(\frac{c_{i+1}-s}{\pi_{i+1}-\pi_{i}}\bigg)ds \bigg]\bigg\rvert_{\alpha_i}^{c_i}\\
&=\bigg[sF^{-1}\bigg(\frac{s-c_{i-1}}{\pi_i-\pi_{i-1}}\bigg)-(\pi_i-\pi_{i-1})G\bigg(F^{-1}\bigg(\frac{s-c_{i-1}}{\pi_{i}-\pi_{i-1}}\bigg)\bigg) \bigg]\bigg\rvert_{\alpha_i}^{c_i}\\
&\qquad\qquad  -\bigg[sF^{-1}\bigg(\frac{c_{i+1}-s}{\pi_{i+1}-\pi_{i}}\bigg)+(\pi_{i+1}-\pi_{i})G\bigg(F^{-1}\bigg(\frac{c_{i+1}-s}{\pi_{i+1}-\pi_{i}}\bigg)\bigg) \bigg]\bigg\rvert_{\alpha_i}^{c_i}\\
&=c_i\bigg[F^{-1}\bigg(\frac{c_i-c_{i-1}}{\pi_i-\pi_{i-1}}\bigg)-F^{-1}\bigg(\frac{c_{i+1}-c_{i}}{\pi_{i+1}-\pi_{i}}\bigg)\bigg]+(\pi_{i+1}-\pi_{i-1})G\bigg(F^{-1}\bigg(\frac{c_{i+1}-c_{i-1}}{\pi_{i+1}-\pi_{i-1}}\bigg)\bigg)\\
& \qquad \qquad -(\pi_{i+1}-\pi_{i})G\bigg(F^{-1}\bigg(\frac{c_{i+1}-c_i}{\pi_{i+1}-\pi_{i}}\bigg)\bigg) -(\pi_{i}-\pi_{i-1})G\bigg(F^{-1}\bigg(\frac{c_{i}-c_{i-1}}{\pi_{i}-\pi_{i-1}}\bigg)\bigg).
\end{align*}
\caption{\label{fig:ex} See Appendix \ref{app:DSIC}. Here, we used integration by parts after the fourth equality.}
\line(1,0){505}
\end{figure*}

\section{Proof of Lemma \ref{lem:ci}}\label{app:ci}
Consider the optimality condition (\ref{eqn:iff}). We start from the system of equations given in (\ref{eqn:optcond}), which when solved backwards starting from $x_N^*$ to $x_1^*$ we get the results in (\ref{eqn:xistar}) and (\ref{eqn:xNstar}). It follows that 
\begin{enumerate}
	\item $x_N^*>0 \Rightarrow \frac{c_N-c_{N-1}}{\pi_N-\pi_{N-1}}>0 \Rightarrow c_N > c_{N-1}$,
	\item $x_i^*>0 \Rightarrow \frac{c_i-c_{i-1}}{\pi_i-\pi_{i-1}} > \frac{c_{i+1}-c_{i}}{\pi_{i+1}-\pi_{i}}$,
\end{enumerate}	
 which is identical to (\ref{eqn:alpha}) after rearrangement.

The converse is proved using backward induction. For $i=N$, we have $c_N=c_{N-1}+(\pi_N-\pi_{N-1})F(x_N^*)$. Hence, if $c_N>c_{N-1}$ we have $F(x_N^*)>0\Rightarrow x_N^*>0.$ Assume that the statement holds for all $i> n$. For $i=n$, by rearranging (\ref{eqn:alpha}), we obtain $\frac{c_i-c_{i-1}}{\pi_i - \pi_{i-1}}>\frac{c_{i+1}-c_i}{\pi_{i+1}- \pi_{i}}$, which, by monotonicity of $F$, implies $ F^{-1}\big(\frac{c_i-c_{i-1}}{\pi_i - \pi_{i-1}}\big)>F^{-1}\big(\frac{c_{i+1}-c_i}{\pi_{i+1}- \pi_{i}}\big)$. This further yields $ x_i^*>0$, by \eqref{eqn:xistar}.

\section{Proof of Lemma \ref{lem:mono}}\label{app:mono}
Note $c_{i-1}\leq \alpha_i< s < \beta_i< c_{i+1}$.  We show $\frac{\partial x_i^*}{\partial s}$ is always positive in the interval $(\alpha_i,\beta_i)$, thereby showing monotonicity of the allocation. First, note that $\frac{\partial F^{-1}(z)}{\partial z} = \frac{1}{f(z)}$. Using this, we get
\begin{align*}
\frac{\partial x_i^*}{\partial s}&= \frac{\partial}{\partial s}F^{-1}\bigg(\frac{s-c_{i-1}}{\pi_i-\pi_{i-1}}\bigg)-\frac{\partial}{\partial s}F^{-1}\bigg(\frac{c_{i+1}-s}{\pi_{i+1}-\pi_{i}}\bigg)\\
&=\frac{1}{(\pi_i-\pi_{i-1})}\frac{1}{f\Big(\frac{s-c_{i-1}}{\pi_i-\pi_{i-1}}\Big)}\\
& \qquad\qquad +\frac{1}{(\pi_{i+1}-\pi_{i})}\frac{1}{f\Big(\frac{c_{i+1}-s}{\pi_{i+1}-\pi_{i}}\Big)} >0,
\end{align*}
which follows from the assumption $\pi_{i-1}<\pi_i<\pi_{i+1}$ and $f(z)>0$ for all $z>0$. Now, for $s\leq \alpha_i$, $x_i^* = 0$ by Lemma \ref{lem:ci} and $s\geq \beta_i$, the optimal allocation is infinite (which we ruled out due to Assumption \ref{assm:bids}. The monotonicity of $x_N^*$ follows as a special case of the proof above. The proof of the lemma is thus complete.


\section{Proof of Theorem \ref{thm:DSIC}} \label{app:DSIC}
Note the following identity:
\beqq{\int_a^b F^{-1}(z) dz &= \int_{F^{-1}(a)}^{F^{-1}(b)} F^{-1}(F(w))f(w) dw \\
& =G(F^{-1}(b)) - G(F^{-1}(a)),}
where we changed the variable $z = F(w)$ and used the invertibility of $F$. Using this identity, we carry out the integral equation for $p_i^*$ for $1\leq i\leq N-1$ in Figure \ref{fig:ex}. The derivation for $p_N^*$ is also analogous.

We next prove the incentive compatibility by showing that the utility of each LSE, $i$, has a unique maximum at the true value $c_i$, regardless of what other buyers announce. Fix $c_{-i}$. The utility of the player is
\begin{align*}
U_i(s,c_{-i})=c_ix_i^*(s,c_{-i})-p_i(s,c_{-i})
\end{align*} 
By differentiating the utility with respect to $s$, we get
\begin{equation*}
\frac{dU_i}{ds}=c_i\frac{dx_i^*	}{ds}-s\frac{dx_i^*}{ds},
\end{equation*}
This derivative is zero when $s=c_i$. Thus, $s=c_i$ satisfies the first order necessary for maximization of the utility function. To show that the utility is indeed maximized at this point, we prove that the utility function is increasing on the interval $(0,c_i)$ and decreasing on $(c_i,\beta_i)$. For $s \in (0,c_i)$,
\begin{equation*}
\frac{dU_i}{ds}=(c_i-s)\frac{dx_i^*}{ds}>0,
\end{equation*}
by Lemma \ref{lem:mono} and $c_i>s$. Therefore, $U_i(\cdot,c_{-i})$ is increasing in this interval. On the other hand, for $s \in (c_i,\beta_i)$,
\begin{equation*}
\frac{dU_i}{ds}=(c_i-s)\frac{dx_i^*}{ds}<0,
\end{equation*}
by Lemma \ref{lem:mono} and $c_i<s$. Thus, $U_i(\cdot,c_{-i})$ is decreasing over this region.

Since $i$ was picked arbitrarily, the result holds for all $i$. Thus, we conclude that it is in the best interest for each LSE to be truthful in their bids. Moreover, truthful bidding strategy is always the best response regardless of what other players bid (as long as the constraints on $c_i$ are met to ensure that the solution to \eqref{eqn:stp} is well-defined).

\begin{figure*}[!t]
\begin{align*}
& p_i^*(c_i,c_{-i})-\ex{\pi_iy_i^*(x^*;\omega)}\\
& = c_ix_i^*+(\pi_{i+1}-\pi_{i-1})G\big(F^{-1}(\rho_{cvx_i})\big)-(\pi_{i+1}-\pi_i)G\big(F^{-1}(\rho_{1_i})\big)-(\pi_i-\pi_{i-1})G\big(F^{-1}(\rho_{2_i})\big)\\
&\quad -\pi_ix_i^*F\big(F^{-1}(\rho_{1_i})\big)-\pi_iF^{-1}(\rho_{2_i})\big[F\big(F^{-1}(\rho_{2_i})\big)-F\big(F^{-1}(\rho_{1_i})\big)\big]+\pi_iG\big(F^{-1}(\rho_{2_i})\big)-\pi_iG\big(F^{-1}(\rho_{1_i})\big),\\
& = (c_i-\pi_i\rho_{1_i})x_i^*+\pi_{i+1}\big[G\big(F^{-1}(\rho_{cvx_i})\big)-G\big(F^{-1}(\rho_{1_i})\big)\big]+\pi_{i-1}\big[G\big(F^{-1}(\rho_{2_i})\big)-G\big(F^{-1}(\rho_{cvx_i})\big)\big]\\
& \quad -\pi_iF^{-1}(\rho_{2_i})(\rho_{2_i}-\rho_{1_i}),\\
& \geq (c_i-\pi_i\rho_{1_i})x_i^*+\pi_{i+1}(\rho_{cvx_i} - \rho_{1_i}) F^{-1}(\rho_{1_i})+\pi_{i-1}(\rho_{2_i} - \rho_{cvx_i})F^{-1}(\rho_{cvx_i})-\pi_iF^{-1}(\rho_{2_i})(\rho_{2_i}-\rho_{1_i})\\
&  =(c_i-\pi_i\rho_{1_i})x_i^* - \pi_i(\rho_{2_i}-\rho_{1_i}) x_i^* + \pi_{i-1}(1-\mu_i)(\rho_{2_i} - \rho_{1_i})(F^{-1}(\rho_{cvx_i}) - F^{-1}(\rho_{1_i})) \\
& = \pi_i\Big(\frac{c_i}{\pi_i} - \rho_{2_i}\Big) x_i^* + \pi_{i-1}\frac{(c_i - \alpha_i)}{\pi_i -\pi_{i-1}}(F^{-1}(\rho_{cvx_i}) - F^{-1}(\rho_{1_i})) .
\end{align*}
\caption{\label{fig:u0bound} See Appendix \ref{app:bound}. Here, the first inequality comes from convexity of $G\circ F^{-1}$.}
\vspace{0.2em}
\line(1,0){505}
\end{figure*}

\section{Proof of Lemma \ref{lem:plv}} \label{app:plv}
In order to prove individual rationality of our mechanism, we need the following lemma.
\begin{lemma}\label{lem:GoF}
$G\circ F^{-1} :[0,1)\rightarrow[0,\infty)$ is a convex monotonically increasing function. 
\end{lemma}
\begin{proof}
We have
\beqq{G(F^{-1}(x)) = \int_0^x F^{-1}(z) dz. }
Differentiating the function once yields
\beqq{\frac{d}{dx} G(F^{-1}(x)) = F^{-1}(x).}
Taking another differentiation, we get
\beqq{\frac{d^2}{dx^2} G(F^{-1}(x)) = \frac{d}{dx}F^{-1}(x) = \frac{1}{f(x)}>0.}
This proves the result.
\end{proof}
Using this lemma, we can show that the sum of the three last terms in \eqref{eqn:pi} is negative.

Note that,
\beq{\frac{c_{i+1}-c_{i-1}}{\pi_{i+1}-\pi_{i-1}}& =\frac{\pi_{i+1}-\pi_{i}}{\pi_{i+1}-\pi_{i-1}}\frac{c_{i+1}-c_i}{\pi_{i+1}-\pi_i}
\nonumber\\
& +\frac{\pi_{i}-\pi_{i-1}}{\pi_{i+1}-\pi_{i-1}}\frac{c_i-c_{i-1}}{\pi_i-\pi_{i-1}}.\label{eq:cvx}}
This expression coupled with the result in Lemma \ref{lem:GoF} yields:
\beqq{& (\pi_{i+1}-\pi_{i-1})G\bigg(F^{-1}\bigg(\frac{c_{i+1}-c_{i-1}}{\pi_{i+1}-\pi_{i-1}}\bigg)\bigg)\\
& \leq (\pi_{i+1}-\pi_{i})G\bigg(F^{-1}\bigg(\frac{c_{i+1}-c_i}{\pi_{i+1}-\pi_{i}}\bigg)\bigg)\nonumber\\
 &\quad +(\pi_{i}-\pi_{i-1})G\bigg(F^{-1}\bigg(\frac{c_{i}-c_{i-1}}{\pi_{i}-\pi_{i-1}}\bigg)\bigg).}
Substituting this inequality in \eqref{eqn:pi} implies,
\beqq{p_i^*(c_i,c_{-i}) \leq c_i x_i^*(c_i,c_{-i}),}
for $1\leq i \leq N-1$. This result is also true for $i=N$ since a positive term is subtracted from $c_Nx_N^*(c_N,c_{-N})$ in \eqref{eqn:pN}.

\section{Proof of Lemma \ref{lem:bound}}\label{app:bound}
We break the process of bounding $U_0^*$ into multiple steps. We start by  first evaluating $\ex{y_i^*(x^*;\omega)}$ in the following lemma.
\begin{lemma}\label{lem:Ey}
The expected value of the optimal recourse variable, $y_i^*(x;\omega)$ is given by
\beq{\ex{y_i^*(x;\omega)}&=x_iF\big(\varphi_{i+1}(x)\big)+\varphi_{i}(x)F\big(\varphi_{i}(x)\big)\nonumber\\
&-\varphi_{i}(x)F\big(\varphi_{i+1}(x)\big)-G\big(\varphi_{i}(x)\big)\nonumber\\
&+G\big(\varphi_{i+1}(x)\big).\label{eq:Ey}}
\end{lemma}
\begin{proof}
By combining \eqref{eq:yistar} and \eqref{eqn:kxomega} we get 
\begin{equation*} 
y_i^*(x;\omega)=  \begin{cases} 
x_i &   w(\omega)< \varphi_{i+1}(x) \\
\varphi_{i}(x)-w(\omega) & \varphi_{i+1}(x) \leq w(\omega) < \varphi_{i}(x) \\
0 & \text{else}
\end{cases},
\end{equation*}
for all $i$, where $\varphi_{N+1}(x)=0$. The expectation with respect to the generation can be computed as follows:
\beqq{\ex{y_i^*(x;\omega)}&=\int_{0}^{\varphi_{i+1}(x)}x_if(\omega)d\omega\\
&+\int_{\varphi_{i+1}(x)}^{\varphi_{i}(x)}(\varphi_{i}(x)-w(\omega))f(\omega)d\omega\\
&=x_iF(\varphi_{i+1}(x))+\varphi_{i}(x)\Big[F(\varphi_{i}(x))\\
&-F(\varphi_{i+1}(x))\Big]-G(\varphi_{i}(x))+G(\varphi_{i+1}(x)).}
The proof of this lemma is complete.
\end{proof}
Before stating our bound for $p_i^*-\pi_i^*\ex{y_i^*(x^*;\omega)}$, we define the variables, $\rho_{1},~\rho_{cvx},~\rho_2$ and $\mu_i$ in the following set of equations,
\beq{\rho_{cvx_i}& =\frac{c_{i+1}-c_{i-1}}{\pi_{i+1}-\pi_{i-1}},~ \rho_{1_i} =\frac{c_{i+1}-c_i}{\pi_{i+1}-\pi_i},\nonumber\\
 \rho_{2_i}& =\frac{c_i-c_{i-1}}{\pi_i-\pi_{i-1}}, 
\quad~\mu_i =\frac{\pi_{i}-\pi_{i-1}}{\pi_{i+1}-\pi_{i-1}}.\label{eq:rhos}} 
where $\rho_{cvx_i}= (1-\mu_i)\rho_{1_i}+\mu_i\rho_{2_i}$, by \eqref{eq:cvx}. We have the following facts:
\begin{align*}
\begin{split}
(\rho_{cvx_i} - \rho_{1_i}) &= \mu_i (\rho_{2_i}-\rho_{1_i})\\
(\rho_{2_i} - \rho_{cvx_i}) &= (1-\mu_i) (\rho_{2_i}-\rho_{1_i})\\
& = \frac{(c_i - \alpha_i)}{\pi_i -\pi_{i-1}}\\
 \pi_i & = \pi_{i+1}\mu_i +\pi_{i-1}(1-\mu_i)
 \end{split}
\end{align*}
which yields
\beq{\label{eqn:facts}\pi_{i+1}(\rho_{cvx_i} - \rho_{1_i})+\pi_{i-1}(\rho_{2_i} - \rho_{cvx_i}) = \pi_i(\rho_{2_i}-\rho_{1_i}).}

For the $N$th LSE, when $F(x)$ is convex over $(0,x^*_N)$, then $F^{-1}(\rho)$ is concave over $(0,\rho_{2_N})$ and the area under $F^{-1}(\rho)$ over this interval can be bounded by the area of the right triangle connecting the points $(0,0), (0,\rho_{2_N}), \text{and} (\rho_{2_N},x_N^*)$. Therefore, by \eqref{eqn:pN} and \eqref{eq:Ey}, we get
\beqq{&p_N^*(c_N,c_{-N})-\ex{\pi_Ny_N^*(x^*;\omega)}=c_Nx^*_N+\pi_{N-1}G(x_N^*)\\
&-\pi_N\frac{c_{N}-c_{N-1}}{\pi_{N}-\pi_{N-1}}x^*_N,\\
&\geq c_Nx^*_N+\frac{1}{2}\pi_{N-1}\frac{c_{N}-c_{N-1}}{\pi_{N}-\pi_{N-1}}x_N^*-\pi_N\frac{c_{N}-c_{N-1}}{\pi_{N}-\pi_{N-1}}x_N^*,\\
& = \pi_N\left(\frac{c_N}{\pi_N} -  \frac{c_{N}-c_{N-1}}{\pi_{N}-\pi_{N-1}} \right)x^*_N+\frac{1}{2}\pi_{N-1}\frac{c_{N}-c_{N-1}}{\pi_{N}-\pi_{N-1}}x_N^*,\\
&=\Big[c_{N-1}\frac{\pi_N}{\pi_N-\pi_{N-1}}-\frac{c_N+c_{N-1}}{2}\frac{\pi_{N-1}}{\pi_N-\pi_{N-1}}\Big]x_N^*.}
The proof of the lemma is complete.
\section{Proof of Lemma \ref{lem:pp}}\label{app:pp}
First, we note that since $c_i>\alpha_i$, we have
\beqq{\frac{c_{i+1}-c_{i-1}}{\pi_{i+1}-\pi_{i-1}} - \frac{c_{i+1}-c_i}{\pi_{i+1}-\pi_i}>0.}
Coupled with the fact that $F^{-1}$ is monotonically increasing function, we conclude that
\beqq{\bigg[F^{-1}\Big(\frac{c_{i+1}-c_{i-1}}{\pi_{i+1}-\pi_{i-1}}\Big)-F^{-1}\Big(\frac{c_{i+1}-c_i}{\pi_{i+1}-\pi_i}\Big)\bigg]>0.}
Now, if $c_i/\pi_i < c_{i-1}/\pi_{i-1}$, then 
\beqq{\frac{c_i}{\pi_i}>\frac{c_i - c_{i-1}}{\pi_i - \pi_{i-1}}\implies c_i-\pi_iF(\varphi_i(x^*)>0,}
where we used \eqref{eqn:optcond}. From Lemma \ref{lem:bound}, we conclude that the lower bound is positive. This further implies that $U_0^*>0$, which completes the proof of the lemma.

\bibliographystyle{ieeetr}
\bibliography{references}

\end{document}